\documentclass[thmsa,12pt]{article}
\normalfont\sffamily
\usepackage{amsfonts}
\usepackage{amssymb}
\usepackage{amsmath}
\usepackage{amsthm}
\usepackage{bm} 
\usepackage{textcomp}
\usepackage{a4wide}
\usepackage{lscape}
\usepackage{url}
\usepackage{xfrac}

\usepackage[normalem]{ulem}
\usepackage[margin=2.4cm]{geometry} 

\usepackage{setspace}
\usepackage{accents}
\usepackage{appendix}

\usepackage{array,booktabs,ragged2e}
\newcolumntype{R}[1]{>{\RaggedRight}p{#1}}

\usepackage{cellspace}
\usepackage{hyperref}
\newtheorem{theorem}{Theorem}
\newtheorem{lemma}{Lemma}

\newtheorem{proposition}{Proposition}

\newtheorem{conjecture}{Conjecture}

\def \be {\begin{equation}}
\def \ee {\end{equation}}
\newcommand{\BZI}{\operatorname{PBnI}}

\begin{document}

\def\title #1{\begin{center}
{\Large {\sc #1}}
\end{center}}
\def\author #1{\begin{center} {#1}
\end{center}}


\title{\sc A note on limit results for the Penrose-Banzhaf index}

\author{Sascha Kurz\\ {\small Dept.\ of Mathematics, University of Bayreuth, Germany, sascha.kurz@uni-bayreuth.de}}


\begin{center} {\bf {\sc Abstract}} \end{center}
\vspace{-4mm}
{\small 
  It is well known that the Penrose-Banzhaf index of a weighted game can differ starkly from corresponding weights. Limit results 
  are quite the opposite, i.e., under certain conditions the power distribution approaches the weight distribution. Here we 
  provide parametric examples that give necessary conditions for the existence of limit results for the Penrose-Banzhaf index. 
}

\noindent
{\small
}
 
\section{Introduction}\label{sec:Introduction}

Consider a private limited company with four shareholders. Assume that the shares are given by $(0.42,0.40,0.09,0.09)$ and that 
decisions are drawn by simple majority rule. The shares suggest that the influence on company decisions is similar for the 
first two and the last two shareowners. However, a proposal can be enforced either by shareholders $2,3,4$ or by shareholder 
$1$ with the support of at least one of the others. Thus, restricting  the analysis to shares and the decision rule, the three later 
shareowners have equal say, which is not reflected by the magnitude of shares at all. In order to evaluate influence in such decision environments, 
power indices like the Penrose-Banzhaf index \cite{banzhaf1964weighted,penrose1952objective}, the Shapley-Shubik index 
\cite{shapley1954method}, or the nucleolus \cite{schmeidler1969nucleolus} were introduced. In our example the 
corresponding power distributions are given by $(\tfrac{1}{2},\tfrac{1}{6},\tfrac{1}{6},\tfrac{1}{6})$, 
$(\tfrac{1}{2},\tfrac{1}{6},\tfrac{1}{6},\tfrac{1}{6})$, and $(\tfrac{2}{5},\tfrac{1}{5},\tfrac{1}{5},\tfrac{1}{5})$, respectively. So, 
our example is just an instance of the well known fact that relative weights can differ starkly from the corresponding power distribution. 
However, under certain conditions, weights and power are almost equal, which is studied under the term \textit{limit results} for 
power indices in the literature. An early example 
was mentioned by Penrose in 1952, see the appendix of \cite{penrose1952objective}. Roughly speaking, if a certain quantity, now 
known as the \textit{Laakso-Taagepera index} \cite{laakso1979effective} or \textit{Herfindahl-Hirschman index}, is large, then 
for simple majority the weights are a good approximation for the Penrose-Banzhaf index. For a specific interpretation of the term 
{\lq\lq}good approximation{\rq\rq}, a proof of some special cases and counter examples have been given in \cite{lindner2004ls} and 
\cite{lindner2007cases}, respectively. Here we study a wider range of measures for deviation and provide parametric examples that 
give necessary conditions for the existence of limit results for the Penrose-Banzhaf index.

We remark that approximations of power indices by relative weights with explicit error bounds, see e.g.\ Theorem~\ref{thm_nuc_approx}, 
are beneficial for several reasons. Since the computation of most known power indices is NP hard, they cannot be computed if the 
number of players gets large, while an approximation may suffice for some applications. Moreover, some weights may be unknown, like 
in a publicly traded stock company, or vary over time, which also prevents the direct computation of the corresponding power distribution.

The remaining part of the paper is structured as follows. After introducing the necessary preliminaries in Section~\ref{sec_preliminaries}, 
we present our main results in Section~\ref{sec_results} and set them into context. Auxiliary results and lengthy proofs are moved to an appendix.

\section{Preliminaries}      
\label{sec_preliminaries}

For a positive integer $n$ let $N=\{1,\dots, n\}$ be the set of players. A \emph{simple game} is a 
mapping $v\colon 2^N\to\{0,1\}$ with $v(\emptyset)=0$, $v(N)=1$, and $v(S)\le v(T)$ for all $S\subseteq T\subseteq N$. 
We call $S\subseteq N\backslash\{i\}$ an \emph{$i$-swing} if $v(S)=0$, $v(S\cup\{i\})=1$ and denote the 
number of $i$-swings in $v$ by $\eta_i(v)$. Setting $\eta(v)=\sum_{i\in N}\eta_i(v)$, the (normalized) Penrose-Banzhaf index of 
player~$i$ is given by $\BZI_i(v)=\eta_i(v)/\eta(v)$. The Shapley-Shubik index is given by the following 
weighted counting of swings
$$
  \operatorname{SSI}_i(v)=\sum_{S\subseteq N\backslash \{i\}} \frac{|S|!\cdot (n-|S|-1)!}{n!} \cdot\left(v(S\cup\{i\})-v(S)\right).
$$ 
The nucleolus $\operatorname{Nuc}(v)$ can be defined as the unique solution of an optimization problem, see \cite{schmeidler1969nucleolus} 
for the details.

A simple game $v$ is weighted if there exists a quota $q\in\mathbb{R}_{>0}$ and weights $w\in\mathbb{R}_{\ge 0}^n$ such that 
$v(S)=1$ iff $w(S)\ge q$, where $w(S):=\sum_{i\in S} w_i$. We write $v=[q;w]$ and speak of \emph{relative} or \textit{normalized weights} if 
$w(N)=1$. By $\Delta(w)=\max\{w_i\,:\,i\in N\}$ and $\Lambda(w)=\max\{w_i/w_j\,:\, i,j\in N, w_i,w_j\neq 0\}$ we denote 
the maximum weight and the span (of weight vector $w$), respectively.

For a vector $x\in\mathbb{R}^n$ we write $\Vert x\Vert_1=\sum_{i=1}^n \left|x_i\right|$, 
$\Vert x\Vert_\infty =\max\left\{\left|x_i\right|\,:\,1\le i\le n\right\}$, and 
$\Vert x\Vert_p=\left(\sum_{i=1}^{n} \left|x_i\right|^p\right)^{1/p}$ for $p\ge 1$. Each such norm $\Vert\cdot\Vert$ induces 
a distance function via $d(x,y)=\Vert x-y\Vert$. For all $x\in\mathbb{R}^n$ and all $1\le p\le p'$ we have 
$\Vert x\Vert_\infty\le \Vert x\Vert_{p'}\le \Vert x\Vert_{p}\le \Vert x\Vert_{1}$, i.e., $\Vert\cdot\Vert_1$ and 
$\Vert\cdot\Vert_\infty$ are the extreme cases on which we focus here. For normalized weight vectors, i.e., $w,w'\in\mathbb{R}^n_{\ge 0}$ 
with $\Vert w\Vert_1=\Vert w'\Vert_1=1$, the inequality $\Vert w-w'\Vert_{\infty}\le\Vert w-w'\Vert_1$ can be strengthened to 
$\Vert w-w'\Vert_{\infty}\le \Vert w-w'\Vert_1/2$, see Lemma~\ref{lemma_improved_relation_infty_1}.  

\section{Approximation results}
\label{sec_results}

Before we start to discuss approximation results between weights and the Penrose-Banzhaf index we briefly review 
the known results for the Shapley-Shubik index and the nucleolus. Neyman's main result of \cite{neyman1982renewal}  
implies as a special case:
\begin{theorem}(Cf.~\cite[Main Theorem]{neyman1982renewal})
  \label{thm_ssi}
  For each $\varepsilon>0$ there exist constants $\delta>0$ and $K>0$ such that for each $n\in\mathbb{N}$, $q\in(0,1)$ and 
  $w\in\mathbb{R}^n$ with $\Vert w\Vert_1=1$, $\Delta(w)<\delta$, and $K\Delta(w)<q<1-K\Delta(w)$, we have 
  $\Vert \operatorname{SSI}([q;w])-w\Vert_1<\varepsilon$. 
\end{theorem}
In words, the Shapley-Shubik index of a weighted game is close to relative weights if the maximum weight is small 
and the quota is not too near to the boundary points $0$ or $1$. We remark that the maximum relative weight is small 
if and only if the Laakso-Taagepera index of $w$ is large, see Lemma~\ref{lemma_relation_maximum_laakso_taagepera} for the 
precise details. Invoking conditions on the maximum weight and the 
quota is indeed necessary for any power index $\varphi$.
\begin{proposition}(\cite[Proposition 1]{weight_polytope})
  \label{prop_impossible}
  Let $\varphi$ be a mapping from the set of weighted games (on $n$ players) into $\mathbb{R}_{\ge 0}^n$.
  \begin{itemize}
    \item[(i)] For each $q\in(0,1]$ and each integer $n\ge 2$ there exists a weighted game $[q;w]$, where $w\in\mathbb{R}^n_{\ge 0}$ and $\Vert w\Vert_1=1$, 
               such that $\Vert w-\varphi([q;w])\Vert_{1}\ge \frac{1}{3}$ and $\Vert w-\varphi([q;w])\Vert_{\infty}\ge \frac{1}{6}$.
    \item[(ii)] For each $\Delta\in(0,1)$ and each integer $n\ge \frac{4}{3\Delta}+6$ there exists a weighted game $[q;w]$, where $q\in(0,1]$, 
                $w\in\mathbb{R}^n_{\ge 0}$, $\Vert w\Vert_1=1$, and $\Delta(w)=\Delta$, 
               such that $\Vert w-\varphi([q;w])\Vert_{1}\ge \frac{1}{3}$, and $\Vert w-\varphi([q;w])\Vert_{\infty}\ge\Delta/4$.
  \end{itemize} 
\end{proposition}
The underlying reason is that different representations of the same weighted game have to be mapped onto the same power vector, i.e., 
the diameter of the polytope of representations of a weighted game plays the key role, as exploited in \cite{weight_polytope}.

While the functional dependence for $\delta$ and $K$ on $\varepsilon$ is hidden in the existence arguments of the proofs of 
\cite{neyman1982renewal}, a more explicit statement for the nucleolus was obtained in \cite{kurz2014nucleolus}:
\begin{theorem}(\cite[Lemma 1]{kurz2014nucleolus})
  \label{thm_nuc_approx}
  For $q\in(0,1)$, $w\in\mathbb{R}^n_{\ge 0}$ with $\Vert w\Vert_1=1$ we have $$\Vert\operatorname{Nuc}([q;w])-w\Vert_1\le 
  \frac{2\Delta(w)}{\min\{q,1-q\}}.$$
\end{theorem} 

For the Penrose-Banzhaf index an analog of Theorem~\ref{thm_ssi} is impossible.
\begin{proposition}
  \label{prop_example_1}
  Let 
  $$
    v_n=[n^3+n^2;2n^2,\overset{2n^3}{\overbrace{1,\dots,1}}].\footnote{The weights of our example look rather artificial and special. 
    However, they are just a specific instance of a much larger class of weighted games. Let $f(n)$, $g(n)$ be two $\mathbb{N}\to\mathbb{N}$  
    functions that tend to infinity as $n$ increases and $c>1$ be an arbitrary constant. If $f(n)\ge (g(n)/f(n))^c$ for all suitably 
    large $n$, then for the sequence of weighted games $v_n=[g(n)+f(n);2f(n),1,\dots,1]$, with $2g(n)$ players of weight $1$, the 
    $\Vert\cdot\Vert_1$-distance between $\BZI(v_n)$ and the corresponding relative weight vector tends to $2$ as $n$ increases. 
    Similarly, the $\Vert\cdot\Vert_\infty$-distance tends to $1$. The proof from the appendix can be easily adopted to that end.}
  $$ 
  Relative weights for a relative quota of $\tfrac{1}{2}$ are given by $w=\left(2n^2,1,\dots,1\right)/(2n^3+2n^2)$, 
  i.e., $\Vert w\Vert_1=1$ and $v_n=[\tfrac{1}{2};w]$. For $n\ge 11$ we have 
  $$
    \Vert \BZI(v_n)-w\Vert_1 \ge 2-\frac{4}{n}
    \quad\text{and}\quad
    \Vert \BZI(v_n)-w\Vert_\infty \ge 1-\frac{2}{n}.
  $$ 
\end{proposition}
\begin{proof}
  See appendix.
\end{proof}
For any given constants $\delta$ and $K$ we can choose $n$ large enough such that $\Delta(v_n)<\delta$ and 
$K\Delta(v_n)<q<1-K\Delta(v_n)$ since $\Delta(v_n)<\tfrac{2}{n}$ and $q=\tfrac{1}{2}$. Thus, 
$\Vert \BZI(v_n)-w\Vert_1 > \varepsilon$ for $\varepsilon<1$ and $n\ge 11$ sufficiently large. 

Note that for any $x,x'\in\mathbb{R}_{\ge 0}^n$ with $\Vert x\Vert_1=\Vert x'\Vert_1=1$ we have $\Vert x-x'\Vert_\infty\le 1$ and 
$\Vert x-x'\Vert_1\le 2$, i.e., for large $n$ the disparity between $\BZI(v_n)$ and the stated relative weights 
is as large as it could be for arbitrary vectors.   

Proposition~\ref{prop_example_1} has an interpretation for the owner structure of a stock company. It models the situation of a main 
holder and a bunch of equivalent small holders. In this form, this example also occurs in many text books. An alternative proof uses the 
normal approximation of the binomial distribution.

For the Penrose-Banzhaf index we have the following limit theorem, see \cite{lindner2004ls}.
\begin{theorem}
\label{thm_bzi} (\cite[Theorem 3.6]{lindner2004ls})
  Let $\widetilde{W}$ be a finite set of non-negative integers, $W\subseteq\widetilde{W}$ be a finite set of positive integers with greatest 
  common divisor $1$, $\rho\in\mathbb{R}_{>0}$, and $\left(w_i\right)_{i\in \mathbb{N}}$ be a sequence with $w_i\in \widetilde{W}$ for all 
  $i\in\mathbb{N}$ such that $\{i\in \mathbb{N}\,:\, w_i\in\widetilde{W}\backslash W\}$ is finite and $\sum_{i\in\{1\le j\le n\,:\,w_j=a\}} w_i\ge \rho\cdot \sum_{1\le i\le n} w_i$ for all $a\in W$ and all 
  sufficiently large $n$. Then, 
  \begin{equation}
    \label{eq_bzi_limit}
    \lim_{n\to\infty} \frac{\BZI_i\left(\left[\tfrac{1}{2},\overline{w}^{(n)}\right]\right)}
    {\BZI_j\left(\left[\tfrac{1}{2},\overline{w}^{(n)}\right]\right)}=\frac{w_i}{w_j}
  \end{equation}
  for all integers $i,j$ with $w_i,w_j\in W$, where $\overline{w}^{(n)}_h=w_h/\sum_{l=1}^n w_l$ denotes the relative weight.
\end{theorem}

Since Equation~(\ref{eq_bzi_limit}) is a statement about the ratio between the Penrose-Banzhaf indices of two players whose weights 
are attained infinitely often, it is trivially satisfied in the example of Proposition~\ref{prop_example_1}. 
So, we give another parametric example. 

\begin{proposition}
  \label{prop_example_2}
  Let 
  $$
    v_n=[3n^3+n^2;\overset{2n+1}{\overbrace{2n^2,\dots,2n^2}},\overset{2n^3}{\overbrace{1,\dots,1}}].\footnote{Again, our example 
    is just a specific instance of a much larger class of weighted games. Let $f(n)$ be an $\mathbb{N}\to\mathbb{N}$  
    function that tends to infinity as $n$ increases and $c>1$ be an arbitrary constant. If $f(n)\ge n^c$ for all suitably 
    large $n$, then for the sequence of weighted games $v_n=[(3n+1)f(n);2f(n),\dots 2f(n),1,\dots,1]$, with $2n+1$ players of weight 
    $2f(n)$ and $2nf(n)$ players of weight $1$, the $\Vert\cdot\Vert_1$-distance between $\BZI(v_n)$ is lower bounded 
    by some positive real number for all suitably large $n$. The proof from the appendix can be easily adopted to that end.}
  $$ 
  Relative weights for a relative quota of $\tfrac{1}{2}$ are given by $w=\left(2n^2,\dots,2n^2,1,\dots,1\right)/(6n^3+2n^2)$, 
  i.e., $\Vert w\Vert_1=1$ and $v_n=[\tfrac{1}{2};w]$. For $n\ge 11$ we have $\BZI_1(v_n)/\BZI_{2n+2}(v_n) 
  \ge 2.6^n/(2n+1)$ and $\Vert \BZI(v_n)-w\Vert_1 \ge \tfrac{1}{5}$.
\end{proposition}
\begin{proof}
  See appendix.
\end{proof}

Here we only have two types of players which we call \textit{large} and \textit{small}. The weight fraction of the small 
players tends to $\tfrac{1}{3}$ as $n$ increases and the maximum relative weight tends to zero. Nevertheless the ratio between 
the Penrose-Banzhaf powers of large and small players grows exponentially faster than the ratio between their weights. Of course this 
does not contradict Theorem~\ref{thm_bzi} since $W$ is assumed to be finite. It was also noted in 
\cite[Section 10]{dubey1979mathematical} that the Penrose-Banzhaf index of $[q;w]$ can behave strangely if the span $\Lambda(w)$ grows 
without bound. To that end we state:
\begin{conjecture}
  \label{conjecture_bzi}
  There exists a constant $C>0$ such that for each $w\in\mathbb{R}_{>0}^n$ with $\Vert w\Vert_1=1$ we have
  $$
    \Vert \BZI([\tfrac{1}{2};w])-w\Vert_1\le C\cdot \Delta(w)\cdot \Lambda(w).
  $$
\end{conjecture}

In Table~\ref{table_conjecture_1} we report some numerical experiments in order to justify Conjecture~\ref{conjecture_bzi}. We consider 
weighted games with two types of players. There are $a$ players of weight $w_1$ and $b$ players of weight $w_2$, where the quota is 
$\left\lceil (aw_1+bw_2)/2\right\rceil$. For each given pair $(w_1,w_2)$ we consider all cases of $2\le a+b\le 1000$ and compute $\Vert\cdot\Vert_1$ 
distance between $\BZI$ and the relative weight distribution. We compute the smallest possible $\sup C\cdot \Lambda(w)$ such that Conjecture~\ref{conjecture_bzi} 
remains true and state the parameters where it is attained.  

From the theoretical side we remark that for most weight vectors $w\in\mathbb{N}^n$ the number of coalitions with a given weight can be relatively good 
approximated by a suitable normal distribution.\footnote{The Frobenius number for two coprime integers $x$ and $y$ is $xy-x-y$ and states the largest number than cannot 
be written as a (non-negative) sum of $x$s and $y$s. So, if the maximum weight is not too large there have to be enough players with weight $x$ and with weight $y$. If there 
are rather few players of some given weight, then those players in total have an inferior impact on the $\Vert\cdot\Vert_1$-distance.} In a small surrounding of the 
expectation\footnote{For a relative quota of $\tfrac{1}{2}$ the expectation is near to $q$.} the curve is rather horizontal -- for simplicity look at ${{2n}\choose {n+i}}$ 
for $i\in\mathbb{Z}$. In the computation of the number of $i$-swings we count coalitions $s$ such that $w(S)<q$ and $w(S\cup\{i\})\ge q$, so that we need this 
flatness between $q-w_i+1$ and $q$, c.f.~Equation~(\ref{eq_eta_1}). Thus, a larger span $\Lambda(w)$ might increase the $\Vert\cdot\Vert_1$-distance for two reasons. 

\begin{table}
  \begin{center}
    \begin{tabular}{rrrrr|rrrrr}
      \hline
      $w_1$ & $w_2$ & $\sup C\cdot \Lambda(w)$ & $a+b$ & $(a,b)$&
      $w_1$ & $w_2$ & $\sup C\cdot \Lambda(w)$ & $a+b$ & $(a,b)$\\ 
      \hline
      1 & 2 & 1.0000 & 2 & $(1,1)$ & 
      1 & 9 & 2.9675 & 999 & $(13,986)$ \\ 
      1 & 3 & 1.3333 & 3 & $(2,1)$ & 
     1 & 11 & 3.5805 & 999 & $(19,980)$ \\ 
      2 & 3 & 1.3333 & 2 & $(1,1)$ & 
     1 & 13 & 4.1971 & 999 & $(27,972)$ \\ 
      1 & 4 & 1.5000 & 4 & $(3,1)$ & 
     1 & 16 & 5.1279 & 1000 & $(41,958)$ \\ 
      3 & 4 & 1.5000 & 2 & $(1,1)$ & 
     1 & 20 & 6.3731 & 999 & $(63,9,36)$ \\ 
      1 & 5 & 1.7973 & 999 & $(3,996)$ & 
     1 & 25 & 7.9279 & 999 & $(97,902)$ \\ 
      2 & 5 & 1.6000 & 3 & $(2,1)$ & 
     2 & 25 & 4.0423 & 999 & $(25,974)$ \\ 
      3 & 5 & 1.5968 & 999 & $(2,997)$ & 
     3 & 25 & 3.0246 & 999 & $(20,979)$ \\ 
      4 & 5 & 1.7988 & 1000 & $(3,997)$ & 
     4 & 25 & 2.5505 & 999 & $(12,987)$ \\ 
      1 & 6 & 2.0794 & 999 & $(5,994)$ & 
     1 & 35 & 11.0181 & 999 & $(191,808)$ \\ %
      5 & 6 & 2.0810 & 1000 & $(5,995)$ & 
     1 & 50 & 15.5376 & 999 & $(375,624)$ \\ 
      \hline
    \end{tabular}
    \caption{Numerical evaluation of Conjecture~\ref{conjecture_bzi}.}
    \label{table_conjecture_1}
  \end{center}
\end{table}

At this place it is appropriate to discuss the relation between approximation errors in the $\Vert\cdot \Vert_1$ norm and 
relative deviations as in Theorem~\ref{thm_bzi}.

\begin{lemma}
  \label{lemma_quotient}
  Let $x,w\in\mathbb{R}_{\ge 0}^n$ with $x_j=x_h$ for all $1\le j,h\le n$ with $w_j=w_h$.
  \begin{enumerate}
    \item[(a)] If $\Vert x-w\Vert_1\le \varepsilon$, then 
    $$
      1-\frac{\varepsilon}{\alpha_i}\le \frac{x_i}{w_i} \le 1+\frac{\varepsilon}{\alpha_i}
    $$
    for all $1\le i\le n$ with $w_i>0$, where $S_i=\{ 1\le j\le n\,:\,w_i=w_j\}$ and $\alpha_i=w(S_i)$.
    \item[(b)] If $w_i,w_j,x_i,x_j\neq 0$, $\varepsilon_i,\varepsilon_j\in[0,1)$ with 
  $1-\varepsilon_i\le\frac{x_i}{w_i}\le 1+\varepsilon _i$ 
  and $1-\varepsilon_j\le\frac{x_j}{w_j}\le 1+\varepsilon_j$, then
  $$
    \frac{1-\varepsilon_i}{1+\varepsilon_j}\le \frac{w_j}{w_i}\cdot \frac{x_i}{x_j}\le \frac{1+\varepsilon_i}{1-\varepsilon_j}
    \quad\text{and}\quad
    \left|\frac{x_i}{w_i}-\frac{x_j}{w_j}\right|\le \varepsilon_i+\varepsilon_j. 
  $$
  \end{enumerate}  
\end{lemma}
\begin{proof}
  Only part (a) is non-trivial. If $x_i/w_i>1+\varepsilon/\alpha_i$ or $x_i/w_i<1-\varepsilon/\alpha_i$ then 
      $$
        \Vert x-w\Vert_1\ge \sum_{j\in S_i} \left|x_j-w_j\right|=|S_i|\cdot \left|x_i-w_i\right|
        > |S_i|\cdot w_i\cdot \varepsilon/\alpha_i=\varepsilon,
      $$
      a contradiction.     
\end{proof}
So, if $\Vert x-w\Vert_1$ is small, then $x_i/w_i$ is near to $1$ and $x_i/x_j$ is near to $w_i/w_j$ provided that the numbers 
are non-zero and the involved players each belong to a family of players with equal weights and non-vanishing weight share. 
Assumptions on $x$, i.e., non-negativity and symmetry, are rather mild and satisfied by any published power index. If true, 
Conjecture~\ref{conjecture_bzi} would imply Theorem~\ref{thm_bzi}. Indeed a small relative deviation is a tighter assumption 
than a small $\Vert\cdot\Vert_1$ distance.
\begin{lemma} 
  \label{lemma_quotient_reverse}
  Let $S\subseteq N=\{1,\dots,n\}$, $\hat{\varepsilon},\tilde{\varepsilon},\varepsilon\in\mathbb{R}_{>0}$, and $x,w\in\mathbb{R}_{\ge 0}^n$
  with $w(N)\le 1$, $w(N\backslash S)\le \hat{\varepsilon}$, $x(N\backslash S)\le \tilde{\varepsilon}$, and 
  $1-\varepsilon\le x_i/w_i\le 1+\varepsilon$ for all $i\in S$, then $\Vert x-w\Vert_1\le \hat{\varepsilon}+\tilde{\varepsilon}+\varepsilon$. 
\end{lemma}
\begin{proof}
  Let $S^+=\{i\in S\,:\, x_i\ge w_i\}$ and $S^-=\{i\in S\,:\,x_i<w_i\}$, then
  \begin{eqnarray*}
    \Vert x-w\Vert_1 &=& \sum_{i\in N}\left|x_i-w_i\right|=\left(x(S^+)-w(S^+)\right)+\left(w(S^-)-x(S^-)\right)\\
    &\le& w(N\backslash S)+x(N\backslash S)+w(S)\cdot\varepsilon\le \hat{\varepsilon}+\tilde{\varepsilon}+\varepsilon.
  \end{eqnarray*}
\end{proof}
In words, if we assume a small relative deviation for all players except a subset of players with a small mass in terms of $x$ and $w$, then   
the $\Vert\cdot\Vert_1$ distance is small.

\medskip

So far we have always assumed a relative quota of $q=\tfrac{1}{2}$ for the Penrose-Banzhaf index. For $q\in(0,1]\backslash\tfrac{1}{2}$ we 
can consider the weighted game $v_{n,q}=[q\cdot 3n;2,\dots,2,1,\dots,1]$ with $n$ players of weight $2$ and another $n$ players of weight $1$.
For each quota $q$ there exists a constant $\varepsilon>0$ such that $\Vert \BZI(v_{n,q})-\overline{w}_{n,q}\Vert_1\ge \varepsilon$ 
for all sufficiently large $n$, where $\overline{w}_{n,q}$ denotes the corresponding relative weight vector, see Proposition~\ref{prop_example_3} 
in the appendix for a more refined statement. Lemma~\ref{lemma_quotient_reverse} implies that the ratio between the Penrose-Banzhaf power of 
players of weight $2$ and players of weight $1$ does not converge to $2$. This example also implies that we cannot have an upper bound 
of the form\\[-3mm]
$$
  \Vert \BZI([q;w])-w\Vert_1\le \frac{C\cdot \Delta(w)^\alpha\cdot \Lambda(w)^\beta}
  {\min\{q,1-q\}^\gamma}
$$
for each $q\in(0,1)$, $w\in\mathbb{R}_{>0}^n$ with $\Vert w\Vert_1=1$, where $C,\alpha,\beta,\gamma\in\mathbb{R}_{>0}$ are arbitrary 
constants. So, there is little room for limit results for the Penrose-Banzhaf index for quotas $q\neq\tfrac{1}{2}$.

With respect to the Shapley-Shubik index we state:
\begin{conjecture}
  \label{conjecture_ssi}
  For each $q\in(0,1)$ and $w\in\mathbb{R}_{>0}^n$ with $\Vert w\Vert_1=1$ we have
  $$
    \Vert \operatorname{SSI}([q;w])-w\Vert_1\le \frac{5\Delta(w)}{\min\{q,1-q\}}.
  $$
\end{conjecture}
We remark that Conjecture~\ref{conjecture_ssi} is valid for all of our three parametric examples as well as 
for all weighted games with at most $9$ players. For the later we have chosen the normalization of a minimum sum integer representation.\footnote{The choice of a specific 
representation is justified by the fact that $\Vert w-w'\Vert_1\le \frac{4\Delta(w)}{\min\{q,1-q\}}$ for two normalized representations of the same weighted 
game, i.e., $[q;w]=[q';w']$, see \cite{weight_polytope}.} For $3\le n\le 9$ the worst case was attained by the game $[n-1;n-1,\underset{n-1}{\underbrace{1,\dots,1}}]$, 
which leads to a $\Vert\cdot\Vert_1$-distance of $\frac{n-2}{n}$ that tends to $1$. 

\section*{Acknowledgements}
I would like to thank the anonymous referee for helpful remarks.


\appendix
\section*{Appendix}

In order to prove Proposition~\ref{prop_example_1} and Proposition~\ref{prop_example_2} we need a small numerical estimate and 
a tightening of the general bound $\Vert x\Vert_\infty\le\Vert x\Vert_1$ in our setting.

\begin{lemma}
  \label{lemma_numerical_1}
  For $n\ge 11$ we have $2n^3/2.6^n\le \tfrac{1}{n}$.
\end{lemma}

\begin{lemma}
  \label{lemma_improved_relation_infty_1}
  For $w,w'\in\mathbb{R}^n_{\ge 0}$ with $\Vert w\Vert_1=\Vert w'\Vert_1=1$, we have 
  $\Vert w-w'\Vert_{\infty}\le \frac{1}{2}\Vert w-w'\Vert_1$.
\end{lemma}
\begin{proof}
  With $S:=\{1\le i\le n\mid w_i\le w'_i\}$ and $A:=\sum_{i\in S} \left(w'_i-w_i\right)$, 
  $B:=\sum_{i\in N\backslash S} \left(w_i-w'_i\right)$, where $N=\{1,\dots,n\}$, we have $A-B=0$ since $\Vert w\Vert_1=\Vert w'\Vert_1$ 
  and $w,w'\in\mathbb{R}_{\ge 0}^n$. Thus, $\Vert w-w'\Vert_1=2A$ and $\Vert w-w'\Vert_{\infty}\le \max\{A,B\}=A$.
\end{proof}

\noindent
\textit{Proof of Proposition~\ref{prop_example_1}.}
  We easily check $\Vert w\Vert_1=1$ and $v_n=[\tfrac{1}{2};w]$. 
  For $v=[q;k,1,\dots,1]$ with $m$ times weight $1$ we have $\eta_1(v)=\sum_{i=1}^{k}{m \choose {q-i}}$ and 
  $\eta_2(v)={{m-1}\choose{q-1}}+{{m-1}\choose{q-k-1}}$ so that
  \begin{equation}
    \label{eq_eta_1}
    \eta_1(v_n) = \sum_{i=1}^{2n^2} {{2n^3}\choose{n^3+n^2-i}}\ge {{2n^3}\choose{n^3}}
  \end{equation}
  and
  \begin{equation}
    \eta_2(v_n) = {{2n^3-1}\choose{n^3+n^2-1}}+{{2n^3-1}\choose{n^3-n^2-1}}
                = {{2n^3-1}\choose{n^3+n^2-1}}+{{2n^3-1}\choose{n^3+n^2}}
                = {{2n^3}\choose{n^3+n^2}} 
  \end{equation} 
  using $q=n^3+n^2$, $k=2n^2$, and $m=2n^3$. 
  
  Since $(1+\tfrac{1}{n})^n$ is monotonically increasing we have $(1+\tfrac{1}{n})^n\ge 2.6$ for $n\ge 11$, so that
  \begin{eqnarray*}
    \frac{\eta_1(v_n)}{\eta_2(v_n)} &\ge & \frac{\left(n^3+n^2\right)!\left(n^3-n^2\right)!}{\left(n^3\right)!\left(n^3\right)!}
    =\frac{\prod\limits_{i=1}^{n^2}n^3+i}{\prod\limits_{i=1}^{n^2}n^3+i-n^2}\ge 
    \left(1+\frac{n^2}{n^3}\right)^{n^2}
    =\left(\left(1+\tfrac{1}{n}\right)^n\right)^n\ge 2.6^n.
  \end{eqnarray*}
  
  From
  $$
    \BZI_1(v_n) = \frac{\eta_1(v_n)}{\eta_1(v_n)+m\cdot \eta_2(v_n)}=1-\frac{m\cdot\eta_2(v_n)}{\eta_1(v_n)+m\cdot \eta_2(v_n)}
    \ge 1-m\cdot \frac{\eta_2(v_n)}{\eta_1(v_n)}\ge 1-\frac{2n^3}{2.6^n},
  $$
  $w_1=\tfrac{1}{n+1}\le\tfrac{1}{n}$, and $2n^3/2.6^n\le \tfrac{1}{n}$ for $n\ge 11$, see Lemma~\ref{lemma_numerical_1}, we deduce
  $$
    \Vert \BZI(v_n)-w\Vert_\infty \ge \left|\BZI_1(v_n)-w_1\right|\ge 1-\tfrac{2}{n}.
  $$
  From Lemma~\ref{lemma_improved_relation_infty_1} we then conclude $\Vert \BZI(v_n)-w\Vert_1\ge 2-\tfrac{4}{n}$.\hfill{$\blacksquare$}

\medskip

\noindent
\textit{Proof of Proposition~\ref{prop_example_2}.}
  We easily check $\Vert w\Vert_1=1$ and $v_n=[\tfrac{1}{2};w]$. For players $1\le i\le 2n+1$ 
  examples of swing coalitions are given by $n$ other players of weight $2n^2$ and $n^3$ players of weight $1$, so that 
  $$
    \eta_i(v_n)\ge {{2n}\choose{n}}\cdot {{2n^3}\choose{n^3}}.
  $$ 
  For players of weight $1$, i.e., $2n+2\le i\le 2n+1+2n^3$, we have
  \begin{eqnarray*}
    \eta_i(v_n)&=&\sum_{j=0}^{2n+1} {{2n+1}\choose j}\cdot{{2n^3-1}\choose{3n^3+n^2 -j\cdot 2n^2-1}}\\ 
    &\le& (n+1)\cdot {{2n+1}\choose n} \cdot{{2n^3-1}\choose{n^3-n^2-1}} +
          (n+1)\cdot {{2n+1}\choose n} \cdot{{2n^3-1}\choose{n^3+n^2-1}}\\
    &=& (n+1)\cdot {{2n+1}\choose n} \cdot{{2n^3}\choose{n^3+n^2}}
    = (2n+1) \cdot {{2n}\choose n} \cdot{{2n^3}\choose{n^3+n^2}}
  \end{eqnarray*}
  Similar as in the proof of Proposition~\ref{prop_example_1} we conclude 
  \begin{equation}
    \label{ie_est_prop_3}
    \frac{\BZI_1(v_n)}{\BZI_{2n+2}(v_n)}=\frac{\eta_1(v_n)}{\eta_{2n+2}(v_n)} \ge \frac{2.6^n}{2n+1},
  \end{equation}   
  noting that $\eta_1(v_n)=\eta_i(v_n)$ for all $1\le i\le 2n+1$ and $\eta_{2n+2}(v_n)=\eta_i(v_n)$ for all $2n+2\le i\le 2n+1+2n^3$ 
  due to symmetry. With this we compute
  \begin{eqnarray*}
    \BZI_1(v_n)-w_1 
    &= & \frac{\eta_1(v_n)}{(2n+1)\cdot\eta_1(v_n)+2n^3\cdot \eta_{2n+2}(v_n)}-\frac{1}{3n+1}\\
    &=& \frac{1}{2n+1}\cdot\left(1-\frac{2n^3\cdot\eta_{2n+2}(v_n)}{(2n+1)\cdot\eta_1(v_n)+2n^3\cdot \eta_{2n+2}(v_n)}\right)-\frac{1}{3n+1}\\
    &\overset{n\ge 3}{\ge}& \frac{1}{2n+1}\cdot\left(\frac{3}{10}-\frac{2n^3}{2n+1}\cdot\frac{\eta_{2n+2}(v_n)}{\eta_{1}(v_n)}\right)\\
    &\overset{\text{(\ref{ie_est_prop_3}})}{\ge}& \frac{1}{2n+1}\cdot\left(\frac{3}{10}-\frac{2n^3}{2.6^n}\right) 
          \overset{\text{Lemma~\ref{lemma_numerical_1}}}{\ge} \frac{1}{2n+1}\cdot\left(\frac{3}{10}-\frac{1}{n}\right)\overset{n\ge 10}{\ge} \frac{1}{2n+1}\cdot\frac{1}{5} 
  \end{eqnarray*}
  for $n\ge 11$ (using Inequality~(\ref{ie_est_prop_3}) and Lemma~\ref{lemma_numerical_1}). Thus
  $$
    \Vert \BZI(v_n)-w\Vert_1\ge \sum_{i=1}^{2n+1}\left|\BZI_i(v_n)-w_i\right|=
    (2n+1)\cdot \left|\BZI_1(v_n)-w_1\right|\ge \frac{1}{5}.\quad\hfill{\blacksquare} 
  $$

\medskip

The details for our briefly sketched last example from Section~\ref{sec_results} are given by:
\begin{proposition}
  \label{prop_example_3}
  For $n\in \mathbb{N}$ and $q\in [0,1]$ let 
  $$
    v_{n,q}=[q\cdot 3n;\overset{n}{\overbrace{2,\dots,2}},\overset{n}{\overbrace{1,\dots,1}}]
  $$ 
  with $n$ times weight $2$ and $n$ times weight $1$. Relative 
  weights for a relative quota of $q$ are given by $w_{n,q}=\left(2,\dots,2,1,\dots,1\right)/(3n)$, 
  i.e., $\Vert w_{n,q}\Vert_1=1$ and $v_{n,q}=[q;w_{n,q}]$. Then the function 
  $f(q):=\lim_{n\to\infty} \Vert \BZI(v_{n,q})-w_{n,q} \Vert_1$ satisfies $f(q)=f(1-q)\in[0,\tfrac{1}{3}]$ 
  for all $q\in[0,1]$ and is strictly monotonically increasing in $[\tfrac{1}{2},1]$. 
\end{proposition}

We refrain from giving a rigorous proof. Symmetry around $q=\tfrac{1}{2}$, i.e., $f(q)=f(1-q)$ follows by 
considering the dual game. For a relative quota $q$ near $0$ or near $1$ all players are equivalent so that 
$\BZI(v_{n,q})=\tfrac{1}{2n}\cdot(1,\dots,1)$, which gives $f(0)=f(1)=\tfrac{1}{3}$. From 
\cite[Theorem 3.6]{lindner2004ls} we conclude $f(\tfrac{1}{2})=0$. In order to check the existence of the limit and monotonicity 
numerically we state
\begin{eqnarray*}
  \eta_{1}(v_{n,q}) &=& \sum_{i=0}^{n-1} {{n-1}\choose{i}}\cdot \left({{n}\choose{\lceil q\cdot 3n\rceil-2i-2}} + {{n}\choose{\lceil q\cdot 3n\rceil-2i-1}}\right)\\
  &=& \sum_{i=0}^{n-1} {{n-1}\choose{i}}\cdot {{n+1}\choose{\lceil q\cdot 3n\rceil-2i-1}}\\
  \eta_{n+1}(v_{n,q}) &=& \sum_{i=0}^n {{n}\choose{i}}\cdot {{n-1}\choose{\lceil q\cdot 3n\rceil-2i-1}}
\end{eqnarray*}
noting that convergence is rather slow and requires high precision computations. We remark that error bounds in general local limit 
theorems for lattice distributions like e.g.\ \cite[Theorem 2 in Chapter VII]{petrov1975sums} are (inevitably) too weak in order to determine $f(q)$ analytically. 
While the summands of $\eta_1(v_{n,q})$ and $\eta_{n+1}(v_{n,q})$ are unimodal and quickly sloping outside a small neighborhood around 
the almost coinciding peaks, the intuitive idea to bound the sums in terms of their maximal summands is not too easy to pursue. The 
maximal summand is not attained at $i\approx q\cdot n$, as one could expect. Even approximating $\log_2 {n\choose k}$ by $n\cdot H(k/n)$, 
where $H(p)=-p\log_2(p)-(1-p)\log_2(1-p)$ is the binary entropy of $p$, gives that the maximum summand is attained for $i\approx n\cdot g(q)$, 
where 
$$
  g(q)=\frac{\tilde{g}(q)^{\tfrac{1}{3}}}{12}-\frac{-3q^2+3q+\tfrac{1}{2}}{\tilde{g}(q)^{\tfrac{1}{3}}}+q
$$   
and
$$
  \tilde{g}(q)=-216q^3+324q^2-108q+6\sqrt{972q^4-1944q^3+864q^2+108q+6}.
$$
Numerically we can check that this fancy function $g$ satisfies $q\le g(q)\le 1.07\cdot g(q)$ for all $\tfrac{1}{2}\le q\le 1$ 
and is, of course, symmetric to $q=\tfrac{1}{2}$.
 
Given the numerical results for $f(q)$ we can state that $\frac{8}{3}\cdot \left|q-\tfrac{1}{2}\right|^3$ and $\tfrac{1}{3}-H(q)/3\log_2(2)$ 
correspond to curves that look similar to $f(q)$ and have a rather small absolute error.

\bigskip
\bigskip

For $w\in\mathbb{R}_{\ge 0}^n$ with $w\neq 0$ the \textit{Laakso-Taagepera} index is given by
$$
    L(w)
    =\left(\sum\limits_{i=1}^{n}w_i\right)^2 / \sum\limits_{i=1}^{n} w_i^2.
$$
In general we have $1\le L(w)\le n$. If the weight vector $w$ is normalized, then the formula simplifies to 
$L(w)=1/\sum_{i=1}^n w_i^2$. Under the name {\lq\lq}effective number of parties{\rq\rq} the index is widely used in political science 
to measure party fragmentation, see, e.g., \cite{laakso1979effective}.  
We observe the following relations between the maximum relative weight 
$\Delta=\Delta(w)$ and the Laakso-Taagepera index $L(w)$:
\begin{lemma} \cite[Lemma 3]{weight_polytope}
  \label{lemma_relation_maximum_laakso_taagepera}
  For $w\in\mathbb{R}_{\ge 0}^n$ with $\Vert w\Vert_1=1$, we have
  $$
    \frac{1}{\Delta}\le
    \frac{1}{\Delta\left(1-\alpha(1-\alpha)\Delta\right)}\le L(w)\le
    \frac{1}{\Delta^2+\frac{(1-\Delta)^2}{n-1}}\le\frac{1}{\Delta^2}
  $$
  for $n\ge 2$, where $\alpha:=\frac{1}{\Delta}-\left\lfloor\frac{1}{\Delta}\right\rfloor\in[0,1)$. 
  If $n=1$, then  $\Delta=L(w)=1$. 
\end{lemma}
\begin{proof}
  The key idea is to optimize $\sum\limits_{i=1}^n w_i^2$ with respect to the constraints $w\in\mathbb{R}^n$, $\Vert w\Vert_1=1$, 
  and $\Delta(w)=\Delta$.

  For $n=1$, we have $w_1=1$, $\Delta(w)=1$, $\alpha=0$, and $L(w)=1$, so that we assume $n\ge 2$ in the remaining part of the proof.
  For $w_i\ge w_j$ consider $a:=\frac{w_i+w_j}{2}$ and $x:=w_i-a$, so that $w_i=a+x$ and $w_j=a-x$. With this we have 
  $w_i^2+w_j^2=2a^2+2x^2$ and $(w_i+y)^2+(w_j-y)^2=2a^2+2(x+y)^2$. Let us assume that $w^\star$ minimizes $\sum_{i=1}^n w_i^2$ under 
  the conditions $w\in\mathbb{R}_{\ge 0}$, $\Vert w\Vert_1=1$, and 
  $\Delta(w)=\Delta$. (Since the target function is continuous and the feasible set is compact and non-empty, a global minimum indeed 
  exists.) W.l.o.g.\ we assume $w_1^\star=\Delta$. If there are indices $2\le i,j\le n$ with $w_i^\star>w_j^\star$, i.e., $x>0$ in the 
  above parameterization, then we may choose $y=-x$. Setting $w_i':=w_i^\star+y=a=\frac{w_i^\star+w_j^\star}{2}$, $w_j':=w_j^\star-y=a
  =\frac{w_i^\star+w_j^\star}{2}$, and $w_h':=w_h^\star$ for all $1\le h\le n$ with $h\notin\{i,j\}$, we have $w'\in \mathbb{R}_{\ge 0}^n$, 
  $\Vert w'\Vert_1=1$, $\Delta(w')=\Delta$, and $\sum_{h=1}^n \left(w_h'\right)^2=\sum_{h=1}^n \left(w_h^\star\right)^2\,-\,x^2$. 
  Since this contradicts the minimality of $w^\star$, we have $w_i^\star=w_j^\star$ for all $2\le i,j\le n$, so that we conclude 
  $w_i^\star=\frac{1-\Delta}{n-1}$ for all $2\le i\le n$ from $1=\Vert w^\star\Vert_1=\sum\limits_{h=1}^n w_h^\star$.
  Thus, $L(w)\le 1/\left(\Delta^2+\frac{(1-\Delta)^2}{n-1}\right)$, which is tight. Since $\Delta\le 1$ and $n\ge 2$, we have 
  $1/\left(\Delta^2+\frac{(1-\Delta)^2}{n-1}\right)\le \frac{1}{\Delta^2}$, which is tight if and only if $\Delta=1$, i.e., 
  $n-1$ of the weights have to be equal to zero.
  
  Now, let us assume that $w$ maximizes $\sum_{i=1}^n w_i^2$ under the conditions $w\in\mathbb{R}_{\ge 0}$, $\Vert w\Vert_1=1$, and 
  $\Delta(w)=\Delta$. (Due to the same reason a global maximum indeed exists.) Due to $1=\Vert w\Vert_1\le n\Delta$ we have $0<\Delta\le 1/n$, 
  where $\Delta=1/n$ implies $w_i=\Delta$ for all $1\le i\le n$. In that case we have $L(w)=n$ and $\alpha=0$, so that the stated lower bounds 
  for $L(w)$ are valid. In the remaining cases we assume $\Delta>1/n$. If there would exist two indices $1\le i,j\le n$ with $w_i\ge w_j$, 
  $w_i<\Delta$, and $w_j>0$, we may strictly increase the target function by moving weight from $w_j$ to $w_i$ (this corresponds to 
  choosing $y>0$), by an amount small enough to still satisfy the constraints $w_i\le \Delta$ and $w_j\ge 0$. Since $\Delta>0$, we can set 
  $a:=\lfloor 1/\Delta\rfloor\ge 0$ with $a\le n-1$ due to $\Delta>1/n$. Thus, for a maximum solution, we
  have exactly $a$ weights that are equal to $\Delta$, one weight that is equal to $1-a\Delta\ge 0$ (which may indeed 
  be equal to zero), and $n-a-1$ weights that are equal to zero. With this and $a\Delta=1-\alpha\Delta$ we have 
  $
    \sum_{i=1}^n w_i^2=a\Delta^2 (1-a\Delta)^2=\Delta-\alpha\Delta^2+\alpha^2\Delta^2=\Delta(1-\alpha\Delta+\alpha^2\Delta)
    =\Delta\left(1-\alpha(1-\alpha)\Delta\right)\le\Delta
  $.   
  Here, the latter inequality is tight if and only if $\alpha=0$, i.e., $1/\Delta\in\mathbb{N}$.
\end{proof}

Let us consider another parametric example.
\begin{lemma}
  Let $v_{n_1,n_2}$ be a weighted game with an even number $n_1\ge 2$ of players of weight $w_1$, an odd number $n_2$ players of weight $1$, and a quota of $q=\tfrac{n_1}{2}w_1+\tfrac{n_2+1}{2}$. If $w_1>n_2$, then the number of swing for a player of weight $w_1$ is given by
  \begin{equation}
    \underset{\frac{1}{2}\cdot {n_1 \choose {n_1/2}}}{\underbrace{{{n_1-1} \choose {n_1/2}}}} \cdot\underset{2^{n_2-1}}{\underbrace{\sum_{i=0}^{\tfrac{n_2-1}{2}} {n_2\choose i}}} \,+\,  \underset{\frac{1}{2}\cdot {n_1 \choose {n_1/2}}}{\underbrace{{{n_1-1} \choose {n_1/2-1}}}} \cdot\underset{2^{n_2-1}}{\underbrace{\sum_{i=\tfrac{n_2+1}{2}}^{n_2} {n_2\choose i}}}= {n_1 \choose {n_1/2}}\cdot 2^{n_2-1}
  \end{equation}
 and the number of swings of a player of weight $1$ is given by
  ${n_1 \choose {n_1/2}}\cdot {{n_2-1} \choose {(n_2-1)/2}}$.
\end{lemma}

\begin{lemma}
  Let $x_{n_1,n_2}:=\left(\tfrac{1}{n_1},\dots,\tfrac{1}{n_1},0,\dots,0\right)$ with $n_1$ positive and $n_2$ zero entries. Then, we have
  \begin{equation}
     \Vert \BZI(v_{n_1,n_2})-x_{n_1,x_2} \Vert_1=2\cdot\frac{n_2\cdot {{n_2-1} \choose {(n_2-1)/2}}/2^{n_2-1}}{n_1+n_2\cdot {{n_2-1} \choose {(n_2-1)/2}}/2^{n_2-1} }
     \approx 2\cdot\frac{\sqrt{\tfrac{2}{\pi}\cdot n_2}}{n_1+\sqrt{\tfrac{2}{\pi}\cdot n_2}}.
  \end{equation}
\end{lemma}
Note that for $w_1\to\infty$ the relative weights of $v_{n_1,n_2}$ tend to $x_{n_1,n_2}$. If $n_1=n^\alpha$ for $0<\alpha<\tfrac{1}{2}$ and $n_2=n-n^\alpha$, then the right hand side tends to $2$ if $n$ tends to infinity, while the maximum relative weight tends to zero and the relative quota tends to $\tfrac{1}{2}$.

\end{document}